\documentclass[format=acmsmall, review=false]{acmart}
\usepackage{acm-ec-26}
\setcopyright{none}
\usepackage{mathtools}
\usepackage{booktabs}

\setcitestyle{authoryear}

\newcommand{\TV}{d_{\mathrm{TV}}}

\begin{document}

\title[Quantifying Theoretical AI Alignment Guarantees]{Quantifying Theoretical AI Alignment Guarantees: Receiver-Utility Bounds in Bayesian Persuasion}
\author{Eric Yachbes}
\affiliation{%
  \institution{Cornell University}
  \city{Ithaca}
  \state{NY}
  \country{United States}
}
\author{Eva Tardos}
\affiliation{%
  \institution{Cornell University}
  \city{Ithaca}
  \state{NY}
  \country{United States}
}

\begin{abstract}
Misalignment can change how information moves from an AI agent to a human user.
We model this as an information advantage: the AI agent observes the world
state, while the human receiver only knows a prior and must act after seeing the
agent's signal. A strategic AI sender may withhold evidence or garble
information in order to steer the human's decision. We ask how much useful
information can still reach the human when the AI optimizes a misaligned
objective. We study a Bayesian persuasion model in which the world state is a
bit string, the human receiver wants to guess the bits correctly, and a single
AI sender wants the receiver to guess as many bits as possible as $1$. For a
prior $\mu$, let $R_0(\mu)$ be the receiver's utility from using only the prior,
and let $R_{\max}(\mu)$ be the largest receiver utility among signaling schemes
that are optimal for the sender. We prove
$R_{\max}(\mu)/R_0(\mu)\leq 3/2$. This bound improves for priors close to the
independent product prior with the same marginals: if
$\mu(x)\geq (1-\eta)\pi_\mu(x)$ for every state $x$, then
$R_{\max}(\mu)\leq R_0(\mu)+\eta n$. We also give a six-bit prior for which
$R_{\max}(\mu)/R_0(\mu)=39/31>5/4$, so no universal $5/4$ bound is possible.
\end{abstract}

\maketitle

\vspace{1cm}
\setcounter{tocdepth}{2}
\tableofcontents

\clearpage

\section{Introduction}

\subsection{Context}

Bayesian persuasion studies a sender who commits to an information policy before
a receiver takes an action \cite{kamenica2011bayesian}. This framework is useful
for AI alignment because an AI agent may have an information advantage over a
human user while also having incentives that differ from the user's incentives.
In our model, the AI is the informed sender: it observes the world state and
chooses what signal to send. The human receiver knows the prior and the
signaling scheme, observes the realized signal, and then acts. A misaligned AI
sender can therefore change the information that reaches the human. We ask how
much receiver utility can remain when the sender chooses an optimal persuasion
scheme for its own objective.

Much work on AI alignment studies misbehavior empirically, for example by
placing models in scenarios where they may mislead users or by training models
to act in more human-aligned ways. These analyses are important, but they do not
by themselves give general guarantees about what a strategic information holder
will reveal. The theoretical role of a persuasion model is to isolate the
mechanism by which incentives affect information transfer.

This paper uses a smaller model than the general persuasion setting. The world
state is a bit string and the receiver's utility is the number of correctly
guessed bits. The simplifying restriction is that the receiver values
coordinatewise accuracy additively. This makes it possible to state numerical
receiver-utility bounds that depend only on the receiver's utility under the
prior.

\subsection{Our Setting and Question}

We study this question in a bit-string model. The human receiver wants to guess
the coordinates of an unknown state correctly, while the AI sender wants the
receiver to guess as many coordinates as possible as $1$. The AI sender observes
the state and commits to a signaling scheme, and the human receiver best
responds to the posterior induced by the realized signal. The signal can be
interpreted as the information, advice, or evidence that the AI agent chooses to
reveal after observing the state.

We compare the receiver's utility under sender-optimal persuasion schemes with
the utility the receiver would obtain from the prior alone. Writing $R_0(\mu)$
for the prior-only utility and $R_{\max}(\mu)$ for the largest receiver utility
among sender-optimal schemes, we ask how large $R_{\max}(\mu)$ can be relative
to $R_0(\mu)$.

\subsection{Contributions}

The main results are the following.

First, Theorem~\ref{thm:universal-upper-bound} proves that, for every prior
$\mu$ on $\{0,1\}^n$, the receiver utility among sender-optimal schemes is at
most
\[
\frac{n+R_0(\mu)}{2}.
\]
Since $R_0(\mu)\geq n/2$, this implies the universal ratio bound
$R_{\max}(\mu)/R_0(\mu)\leq 3/2$. Thus, even when the sender chooses an
optimal persuasion scheme for its own objective, the receiver's utility cannot
exceed its prior utility by more than this universal factor.

Second, product and near-product priors admit sharper bounds.
Theorem~\ref{thm:product-prior-sender-value} and
Corollary~\ref{cor:product-no-receiver-gain} show that independent product
priors have no receiver gain over the prior benchmark. More generally,
Corollary~\ref{cor:pointwise-domination-stability} shows that if $\pi_\mu$
denotes the product prior with the same marginals as $\mu$ and
\[
\mu(x)\geq (1-\eta)\pi_\mu(x)
\qquad\text{for every }x\in\{0,1\}^n,
\]
then
\[
R_{\max}(\mu)\leq R_0(\mu)+\eta n,
\]
and therefore $R_{\max}(\mu)/R_0(\mu)\leq 1+2\eta$. This shows that weak
dependence among coordinates limits how much sender-optimal persuasion can
increase receiver utility over the prior benchmark.

Third, Theorem~\ref{thm:six-bit-lower-bound} gives a six-bit prior with
\[
\frac{R_{\max}(\mu)}{R_0(\mu)}=\frac{39}{31}> \frac54.
\]
Thus the constant $5/4$ is false as a universal receiver-utility bound:
correlations in the prior can create larger receiver gains than product-like
instances allow.

\section{Related Work}

Kamenica and Gentzkow introduced Bayesian persuasion as a model in which an
informed sender commits to an information structure before a receiver acts
\cite{kamenica2011bayesian}. Their central reduction is to view a signaling
scheme as a Bayes-plausible distribution over posterior beliefs, so that the
sender's problem becomes a concavification problem over receiver posteriors. We
keep the same commitment timing and Bayesian receiver response, but specialize
the action to a bit string and measure receiver welfare by coordinatewise
accuracy. This specialization lets us prove a receiver-utility ratio bound that
does not require solving the sender's full concavification problem.

The closest alignment motivation is Collina et al.'s model of emergent
alignment via competition \cite{collina2025emergent}. They consider a human
interacting with several differently misaligned AI agents in a multi-leader
Stackelberg game extending Bayesian persuasion to conversations. Their main
condition is geometric: if the human utility is approximately in the convex hull
of the agents' utilities, then in equilibrium competition can give the human
utility comparable to interacting with a perfectly aligned agent. This gives a
positive guarantee when model diversity is strong enough to satisfy the convex
hull condition. Our paper studies a complementary case: there is only one
misaligned sender, so there is no competitive force and no convex-hull
assumption. The question becomes how much receiver utility can remain solely
from the sender's optimal information policy.

Multi-sender persuasion shows both why competition can help and why broad
equilibrium statements are delicate. Gentzkow and Kamenica study symmetric
information games in which senders choose information structures and identify
conditions under which competition yields outcomes at least as informative as
collusion \cite{gentzkow2017competition}. Their model allows rich signal spaces,
including information structures that can be arbitrarily correlated across
senders. Hossain et al. instead study simultaneous independent senders from a
computational perspective and show that even computing a sender's best response
is NP-hard and that finding Nash equilibria is PPAD-hard in general
\cite{hossain2024multisender}. These results motivate working in a restricted
bit-guessing model where exact receiver-utility comparisons can be proved.

Haifeng Xu's work with Dughmi and subsequent work on public persuasion are
especially close to our setting. In public persuasion with no externalities,
there are many receivers, each receiver has a binary action, all receivers see
the same public signal, and each receiver's payoff depends on the state and on
her own action but not on the actions of other receivers. Our single receiver's
bitwise decision problem can be viewed in the same way: coordinate $i$ is a
binary-action receiver, the common signal is the sender's message, and the
coordinate chooses $1$ exactly when the posterior marginal of $X_i$ is at least
$1/2$. Under this identification, our obedience constraints are the usual public
persuasion obedience constraints for multiple binary receivers with no
externalities.

Dughmi and Xu separate private from public signaling in this model: private
signaling admits broad positive algorithmic results, while public signaling can
be hard to approximate even for additive objectives
\cite{dughmi2017algorithmic}. Xu later shows that public persuasion in the same
no-externalities model is fixed parameter tractable by using a geometric
characterization based on hyperplane arrangements \cite{xu2020tractability}.
Our objective is a particularly structured public-persuasion objective: the
sender wants to maximize the number of coordinates choosing $1$, while the
receiver's welfare is the number of correct coordinate actions. We use this
additional structure to prove a closed-form receiver-welfare bound rather than
an algorithm for computing an optimal public scheme.

Finally, this paper uses the primal-dual style common in recent work on
persuasion. Dworczak and Martini interpret optimal persuasion through prices
\cite{dworczak2019simple}, and Dworczak and Kolotilin give a general duality
theory in which dual variables certify optimality of persuasion schemes
\cite{dworczak2024duality}. The six-bit example below follows this certificate
logic: an explicit obedient scheme proves attainability, and an explicit dual
certificate proves that the sender payoff cannot exceed $5$.

\section{Model}

\subsection{States, Actions, and Payoffs}

Let $n\geq 1$. Let $X\in\{0,1\}^n$ be the world state, distributed according to
a common prior $\mu$. The receiver chooses an action $a\in\{0,1\}^n$. The
receiver's payoff is the number of coordinates guessed correctly,
\[
u_R(x,a)=\sum_{i=1}^n \mathbf{1}\{a_i=x_i\},
\]
and the sender's payoff is the number of coordinates guessed as $1$,
\[
u_S(a)=\sum_{i=1}^n a_i.
\]
The sender observes $X$ and commits to a signaling scheme. The receiver observes
the realized signal and best responds to the posterior. We use weak obedience:
if the posterior marginal of $X_i$ is exactly $1/2$ for a coordinate $i$, either
bit value may be selected for that coordinate.

\subsection{Direct Obedient Form}

Let $\mu$ be a prior on $\{0,1\}^n$. A signaling scheme can be written in direct
form as a joint distribution $q(x,a)$ over states $x\in\{0,1\}^n$ and
recommended actions $a\in\{0,1\}^n$ such that
\[
q(x,a)\geq 0,
\qquad
\sum_{a\in\{0,1\}^n}q(x,a)=\mu(x)
\quad\text{for every }x.
\]
Direct recommendations are without loss for the utilities in this paper: after
each signal, record the receiver action induced by that signal.

For each action $a$ and coordinate $i$, define
\[
C_i(x,a)=
\begin{cases}
\frac12-x_i, & a_i=1,\\
x_i-\frac12, & a_i=0.
\end{cases}
\]
The recommendation $a$ is obedient on coordinate $i$ exactly when
\[
\sum_x q(x,a) C_i(x,a)\leq 0.
\]
If $a_i=1$, this inequality says that the posterior probability of $X_i=1$ is
at least $1/2$. If $a_i=0$, it says that the posterior probability of $X_i=1$
is at most $1/2$.
We call $q$ an obedient direct scheme if these inequalities hold for every
recommended action $a$ and coordinate $i$.

\subsection{Utilities}

For an obedient direct scheme $q$, define
\[
R(q)=\sum_{x,a}q(x,a)\sum_{i=1}^n \mathbf{1}\{a_i=x_i\},
\qquad
S(q)=\sum_{x,a}q(x,a)\sum_{i=1}^n a_i.
\]
The sender's optimal value is
\[
S^*(\mu)=\max_{q\text{ an obedient direct scheme}} S(q).
\]
The receiver's best utility among sender-optimal obedient direct schemes is
\[
R_{\max}(\mu)
=
\max\{R(q):q\text{ an obedient direct scheme and }S(q)=S^*(\mu)\}.
\]

For any prior $D$ on $\{0,1\}^n$, write
\[
p_i(D)=\Pr_D[X_i=1],
\qquad
P(D)=\sum_{i=1}^n p_i(D),
\qquad
M(D)=\sum_{i=1}^n \min\{1,2p_i(D)\}.
\]
Here $P(D)$ is the prior expected number of $1$ coordinates in the state.
The quantity $M(D)$ is the coordinatewise upper envelope for sender utility:
on coordinate $i$, obedience can allow recommendation $1$ with probability at
most $1$ when $p_i(D)\geq 1/2$, and with probability at most $2p_i(D)$ when
$p_i(D)<1/2$.
The receiver's prior utility is
\[
R_0(D)=\sum_{i=1}^n \max\{p_i(D),1-p_i(D)\}.
\]

\section{Results}

\subsection{Basic Accounting}

\begin{lemma}[Prior utility identity]
\label{lem:prior-utility-identity}
For every prior $\mu$ on $\{0,1\}^n$,
\[
R_0(\mu)=n+P(\mu)-M(\mu).
\]
\end{lemma}

\begin{proof}
For a coordinate with $p_i(\mu)\leq 1/2$,
\[
\max\{p_i(\mu),1-p_i(\mu)\}
=1-p_i(\mu)=1+p_i(\mu)-2p_i(\mu).
\]
For a coordinate with $p_i(\mu)\geq 1/2$,
\[
\max\{p_i(\mu),1-p_i(\mu)\}
=p_i(\mu)=1+p_i(\mu)-1.
\]
Summing over coordinates gives the identity.
\end{proof}

\begin{lemma}[Sender-receiver tradeoff]
\label{lem:sender-receiver-tradeoff}
For every prior $\mu$ and every obedient direct scheme $q$,
\[
R(q)+S(q)\leq n+P(\mu).
\]
\end{lemma}

\begin{proof}
For each state-action pair $(x,a)$ and coordinate $i$,
\[
\mathbf{1}\{a_i=x_i\}+a_i\leq 1+x_i.
\]
Summing over coordinates and taking expectations under $q$ gives
\[
R(q)+S(q)\leq n+\sum_{i=1}^n \Pr_\mu[X_i=1]=n+P(\mu).
\]
\end{proof}

\begin{lemma}[Coordinate upper bound for sender utility]
\label{lem:sender-coordinate-upper-bound}
For every prior $\mu$,
\[
S^*(\mu)\leq M(\mu).
\]
\end{lemma}

\begin{proof}
Fix an obedient direct scheme $q$, and let $A$ be the recommended action random
variable under $q$. For coordinate $i$,
\[
S(q)=\sum_{i=1}^n \Pr[A_i=1].
\]
If $p_i(\mu)\geq 1/2$, then
\[
\Pr[A_i=1]\leq 1=\min\{1,2p_i(\mu)\}.
\]
If $p_i(\mu)<1/2$, the obedience constraints for recommended actions with
$a_i=1$ give
\[
\sum_{a:a_i=1}\sum_x q(x,a)\left(x_i-\frac12\right)\geq 0.
\]
Equivalently,
\[
\frac12\Pr[A_i=1]\leq \Pr[X_i=1,A_i=1]\leq p_i(\mu),
\]
so $\Pr[A_i=1]\leq 2p_i(\mu)$. Summing over coordinates gives
\[
S(q)\leq M(\mu).
\]
Taking the maximum over obedient direct schemes proves the lemma.
\end{proof}

\subsection{Universal Upper Bound}

\begin{lemma}[OR-scheme lower bound]
\label{lem:or-scheme-lower-bound}
For every prior $\mu$,
\[
S^*(\mu)\geq \frac{M(\mu)+P(\mu)}{2}.
\]
\end{lemma}

\begin{proof}
Construct a signaling scheme whose induced direct scheme is
$q^{\mathrm{OR}}$. After observing $X$, let the sender draw an independent state
$Y\sim\mu$ and send the unordered pair $\{X,Y\}$. If $X\neq Y$, the posterior is
uniform on the two states in the pair. If $X=Y$, the posterior is degenerate.
The receiver's recommended action is
\[
a=X\vee Y,
\]
where $\vee$ denotes bitwise OR.

The induced direct scheme $q^{\mathrm{OR}}$ is obedient. If $a_i=1$, then at
least one of $X_i,Y_i$ is $1$, so the posterior marginal of $X_i$ is either $1$
or $1/2$. If $a_i=0$, both states have $0$ in coordinate $i$, so the posterior
marginal of $X_i$ is $0$.

The sender utility under $q^{\mathrm{OR}}$ is
\[
\begin{aligned}
S(q^{\mathrm{OR}})
&=\mathbb{E}[|X\vee Y|_1]\\
&=\sum_{i=1}^n \Pr[X_i=1\text{ or }Y_i=1]\\
&=\sum_{i=1}^n \left(1-(1-p_i(\mu))^2\right)
=\sum_{i=1}^n (2p_i(\mu)-p_i(\mu)^2).
\end{aligned}
\]
For every $p\in[0,1]$,
\[
2p-p^2\geq \frac{\min\{1,2p\}+p}{2}.
\]
If $p\leq 1/2$, this is $2p-p^2\geq 3p/2$. If $p\geq 1/2$, this is
\[
2p-p^2\geq \frac{1+p}{2},
\]
equivalently $(1-p)(2p-1)\geq 0$. Summing over coordinates gives
\[
S(q^{\mathrm{OR}})\geq \frac{M(\mu)+P(\mu)}{2}.
\]
Since $q^{\mathrm{OR}}$ is obedient, the same lower bound holds for
$S^*(\mu)$.
\end{proof}

\begin{theorem}[Universal upper bound]
\label{thm:universal-upper-bound}
For every prior $\mu$ on $\{0,1\}^n$,
\[
R_{\max}(\mu)\leq \frac{n+R_0(\mu)}{2}.
\]
Consequently,
\[
\frac{R_{\max}(\mu)}{R_0(\mu)}\leq \frac32.
\]
\end{theorem}

\begin{proof}
Let $q$ be any sender-optimal obedient direct scheme. By the sender-receiver
tradeoff, Lemma~\ref{lem:sender-receiver-tradeoff}, we have
$R(q)+S(q)\leq n+P(\mu)$. Since $q$ is sender-optimal, $S(q)=S^*(\mu)$. The
OR-scheme lower bound, Lemma~\ref{lem:or-scheme-lower-bound}, gives
$S^*(\mu)\geq (M(\mu)+P(\mu))/2$. Therefore
\[
\begin{aligned}
R(q)
&\leq n+P(\mu)-S^*(\mu)\\
&\leq n+P(\mu)-\frac{M(\mu)+P(\mu)}{2}
=n+\frac{P(\mu)}{2}-\frac{M(\mu)}{2}.
\end{aligned}
\]
Using the prior utility identity, Lemma~\ref{lem:prior-utility-identity},
$R_0(\mu)=n+P(\mu)-M(\mu)$, so this becomes
\[
R(q)\leq R_0(\mu)+\frac{M(\mu)-P(\mu)}{2}.
\]
The definitions of $M(\mu)$ and $P(\mu)$ give
\[
M(\mu)-P(\mu)
=\sum_{i=1}^n \min\{p_i(\mu),1-p_i(\mu)\}
=n-R_0(\mu).
\]
Therefore
\[
R(q)\leq R_0(\mu)+\frac{n-R_0(\mu)}{2}
=\frac{n+R_0(\mu)}{2}.
\]
Taking the maximum over sender-optimal obedient direct schemes gives the stated
upper bound on $R_{\max}(\mu)$. Since $R_0(\mu)\geq n/2$, the ratio bound
follows.
\end{proof}

\subsection{Product and Near-Product Priors}

For a prior $\mu$, let
\[
\pi_\mu=\bigotimes_{i=1}^n \mathrm{Bernoulli}(p_i(\mu))
\]
be the product prior with the same marginals as $\mu$.

\begin{theorem}[Product-prior sender value]
\label{thm:product-prior-sender-value}
For every prior $\mu$,
\[
S^*(\pi_\mu)=M(\mu).
\]
\end{theorem}

\begin{proof}
Since $\pi_\mu$ and $\mu$ have the same marginals, $M(\pi_\mu)=M(\mu)$.
The coordinate upper bound, Lemma~\ref{lem:sender-coordinate-upper-bound},
gives $S^*(D)\leq M(D)$ for every prior $D$, and hence
\[
S^*(\pi_\mu)\leq M(\pi_\mu)=M(\mu).
\]

It remains to construct an obedient direct scheme for $\pi_\mu$ with sender
utility $M(\mu)$. Independently across coordinates, use the following bitwise
recommendation rule. If $p_i(\mu)\geq 1/2$, always recommend $A_i=1$. If
$p_i(\mu)<1/2$, recommend $A_i=1$ with probability $1$ when $X_i=1$, and
recommend $A_i=1$ with probability $p_i(\mu)/(1-p_i(\mu))$ when $X_i=0$.
Otherwise recommend $A_i=0$.

Under this rule, when $p_i(\mu)<1/2$ and recommended action $A_i=1$ has
positive probability, the posterior marginal of $X_i$ is exactly $1/2$; the
posterior marginal of $X_i$ after recommended action $A_i=0$ is $0$. Hence the
bitwise recommendation rule is obedient under weak obedience.

Because $\pi_\mu$ is a product prior and the recommendation rule is independent
across coordinates conditional on $X$, the posterior marginal of $X_i$ after
observing the full recommended action $A$ depends only on $A_i$. Therefore the
full direct scheme is obedient and has sender utility
\[
\sum_{i:p_i(\mu)\geq 1/2}1+\sum_{i:p_i(\mu)<1/2}2p_i(\mu)
=M(\mu).
\]
Thus $S^*(\pi_\mu)=M(\mu)$.
\end{proof}

\begin{lemma}[Sender-value deficit implies receiver bound]
\label{lem:sender-value-deficit}
If a prior $\mu$ satisfies
\[
S^*(\mu)\geq M(\mu)-\Delta,
\]
then
\[
R_{\max}(\mu)\leq R_0(\mu)+\Delta.
\]
\end{lemma}

\begin{proof}
Let $q$ be any sender-optimal obedient direct scheme for $\mu$. The
sender-receiver tradeoff gives
\[
R(q)\leq n+P(\mu)-S^*(\mu).
\]
Using $S^*(\mu)\geq M(\mu)-\Delta$ and the prior utility identity,
Lemma~\ref{lem:prior-utility-identity}, which gives
$R_0(\mu)=n+P(\mu)-M(\mu)$,
\[
R(q)\leq n+P(\mu)-M(\mu)+\Delta=R_0(\mu)+\Delta.
\]
Taking the maximum over sender-optimal obedient direct schemes gives the bound
on $R_{\max}(\mu)$.
\end{proof}

\begin{corollary}[Product priors give no receiver gain]
\label{cor:product-no-receiver-gain}
For every prior $\mu$,
\[
R_{\max}(\pi_\mu)=R_0(\pi_\mu).
\]
\end{corollary}

\begin{proof}
The product-prior sender value theorem,
Theorem~\ref{thm:product-prior-sender-value}, gives
$S^*(\pi_\mu)=M(\mu)$. Therefore the sender-value deficit lemma,
Lemma~\ref{lem:sender-value-deficit}, applies with $\Delta=0$ and gives
$R_{\max}(\pi_\mu)\leq R_0(\pi_\mu)$.

The product-prior construction in the proof of the previous theorem is
sender-optimal. Its receiver utility equals $R_0(\pi_\mu)$ coordinate by
coordinate. If $p_i(\mu)\geq 1/2$, the scheme always recommends $A_i=1$, giving
receiver utility $p_i(\mu)$. If $p_i(\mu)<1/2$, recommended action $A_i=1$
contributes correctness probability $p_i(\mu)$ and recommended action $A_i=0$
contributes correctness probability $1-2p_i(\mu)$, for total
$1-p_i(\mu)$. Summing over coordinates gives receiver utility $R_0(\pi_\mu)$,
so $R_{\max}(\pi_\mu)\geq R_0(\pi_\mu)$.
\end{proof}

\begin{lemma}[Sender-value Lipschitz bound on $K_\delta$]
\label{lem:sender-value-lipschitz}
Let $\delta>0$ and
\[
K_\delta=\{D\in\Delta(\{0,1\}^n):D(x)\geq \delta\text{ for all }x\}.
\]
For any $D,E\in K_\delta$,
\[
|S^*(D)-S^*(E)|
\leq
\frac{2n}{\delta}\TV(D,E).
\]
\end{lemma}

\begin{proof}
For the sender LP defining $S^*(D)$, take free dual variables $\alpha_x$ for
the Bayes constraints and nonnegative dual variables $\lambda_{a,i}$ for the
obedience constraints. The dual program is
\[
\min_{\alpha,\lambda} \sum_x D(x)\alpha_x
\]
subject to $\lambda_{a,i}\geq 0$ and
\[
\alpha_x-\sum_{i=1}^n (2a_i-1)\left(x_i-\frac12\right)\lambda_{a,i}
\geq |a|_1
\qquad\text{for every }x,a\in\{0,1\}^n.
\]
The dual feasible set does not depend on $D$. If $D\in K_\delta$ and
$(\alpha,\lambda)$ is an optimal dual solution, then setting $a=x$ gives
\[
\alpha_x-\frac12\sum_i\lambda_{x,i}\geq |x|_1,
\]
so $\alpha_x\geq 0$. Since $S^*(D)\leq n$ and
\[
\sum_x D(x)\alpha_x=S^*(D),
\]
we have $\delta\alpha_x\leq n$ for every $x$, and hence
$\|\alpha\|_\infty\leq n/\delta$.

For any $D,E\in K_\delta$, cross-evaluating optimal dual solutions for $D$ and
$E$ gives
\[
|S^*(D)-S^*(E)|
\leq
\frac{n}{\delta}\|D-E\|_1
=
\frac{2n}{\delta}\TV(D,E).
\]
\end{proof}

\begin{corollary}[Total-variation stability near product priors]
\label{cor:tv-stability}
If $\delta>0$ and $\mu,\pi_\mu\in K_\delta$, then
\[
R_{\max}(\mu)
\leq
R_0(\mu)+\frac{2n}{\delta}\TV(\mu,\pi_\mu).
\]
\end{corollary}

\begin{proof}
The sender-value Lipschitz bound on $K_\delta$,
Lemma~\ref{lem:sender-value-lipschitz}, gives
$|S^*(\mu)-S^*(\pi_\mu)|
\leq (2n/\delta)\TV(\mu,\pi_\mu)$, and hence
\[
S^*(\mu)\geq S^*(\pi_\mu)-\frac{2n}{\delta}\TV(\mu,\pi_\mu).
\]
By the product-prior sender value theorem,
Theorem~\ref{thm:product-prior-sender-value}, $S^*(\pi_\mu)=M(\mu)$. Applying
the sender-value deficit lemma, Lemma~\ref{lem:sender-value-deficit}, with
\[
\Delta=\frac{2n}{\delta}\TV(\mu,\pi_\mu)
\]
proves the corollary.
\end{proof}

\begin{lemma}[Mixture lower bound for sender value]
\label{lem:mixture-lower-bound}
Suppose that, for some $\eta\in[0,1]$,
\[
\mu(x)\geq (1-\eta)\pi_\mu(x)
\qquad\text{for every }x\in\{0,1\}^n.
\]
Then
\[
S^*(\mu)\geq (1-\eta)M(\mu).
\]
\end{lemma}

\begin{proof}
If $\eta=0$, then $\mu=\pi_\mu$ and the claim follows from
Theorem~\ref{thm:product-prior-sender-value}, which gives
$S^*(\pi_\mu)=M(\mu)$. If $\eta>0$, define the residual prior
\[
H(x)=\frac{\mu(x)-(1-\eta)\pi_\mu(x)}{\eta}.
\]
Then
\[
\mu=(1-\eta)\pi_\mu+\eta H.
\]
Let $q^\Pi$ be a sender-optimal obedient direct scheme for $\pi_\mu$, so
$S(q^\Pi)=M(\mu)$. Let $q^H$ be any obedient direct scheme for $H$. The mixture
\[
q=(1-\eta)q^\Pi+\eta q^H
\]
is an obedient direct scheme for $\mu$ because the marginal constraints and
obedience constraints are linear in $q$. Its sender utility satisfies
\[
S(q)=(1-\eta)S(q^\Pi)+\eta S(q^H)\geq (1-\eta)M(\mu).
\]
Therefore $S^*(\mu)\geq (1-\eta)M(\mu)$.
\end{proof}

\begin{corollary}[Pointwise-domination stability near product priors]
\label{cor:pointwise-domination-stability}
Suppose that, for some $\eta\in[0,1]$,
\[
\mu(x)\geq (1-\eta)\pi_\mu(x)
\qquad\text{for every }x\in\{0,1\}^n.
\]
Then
\[
R_{\max}(\mu)\leq R_0(\mu)+\eta M(\mu)\leq R_0(\mu)+\eta n.
\]
Consequently,
\[
\frac{R_{\max}(\mu)}{R_0(\mu)}
\leq
1+\eta\frac{M(\mu)}{R_0(\mu)}
\leq
1+2\eta.
\]
\end{corollary}

\begin{proof}
The mixture lower bound, Lemma~\ref{lem:mixture-lower-bound}, gives
$S^*(\mu)\geq (1-\eta)M(\mu)$, equivalently
\[
S^*(\mu)\geq M(\mu)-\eta M(\mu).
\]
The sender-value deficit lemma, Lemma~\ref{lem:sender-value-deficit}, converts
this lower bound on $S^*(\mu)$ into the receiver bound with
$\Delta=\eta M(\mu)$:
\[
R_{\max}(\mu)\leq R_0(\mu)+\eta M(\mu).
\]
The additive bound follows from $M(\mu)\leq n$. The ratio bound follows from
$M(\mu)\leq n$ and $R_0(\mu)\geq n/2$.
\end{proof}

\subsection{A Six-Bit Lower Bound}

For this subsection, let $\mu^{\mathrm{lb}}$ be the prior on $\{0,1\}^6$ given
by
\[
\begin{aligned}
\mu^{\mathrm{lb}}(010111)&=\frac{3}{10},&
\mu^{\mathrm{lb}}(100110)&=\frac{1}{10},&
\mu^{\mathrm{lb}}(100011)&=\frac{1}{10},\\
\mu^{\mathrm{lb}}(001000)&=\frac{1}{5},&
\mu^{\mathrm{lb}}(111000)&=\frac{1}{5},&
\mu^{\mathrm{lb}}(100101)&=\frac{1}{10}.
\end{aligned}
\]

\begin{lemma}[Six-bit primal scheme]
\label{lem:six-bit-primal}
For the prior $\mu^{\mathrm{lb}}$, there is an obedient direct scheme $q$ with
\[
S(q)=5
\qquad\text{and}\qquad
R(q)=\frac{39}{10}.
\]
\end{lemma}

\begin{proof}
Consider the following direct scheme $q$. All omitted entries have mass zero.
\[
\begin{array}{c|c}
\text{recommended action }a & \text{state masses assigned to }a\\
\hline
100111 & 100101:\frac{1}{10},\;100110:\frac{1}{10}\\
011111 & 001000:\frac{1}{5},\;010111:\frac{1}{5}\\
111011 & 100011:\frac{1}{10},\;111000:\frac{1}{10}\\
111111 & 010111:\frac{1}{10},\;111000:\frac{1}{10}
\end{array}
\]
The posterior marginals of $X$ conditional on each recommended action in the
support of $q$ are
\[
\begin{array}{c|c}
a & \text{posterior marginals}\\
\hline
100111 & (1,0,0,1,\frac12,\frac12)\\
011111 & (0,\frac12,\frac12,\frac12,\frac12,\frac12)\\
111011 & (1,\frac12,\frac12,0,\frac12,\frac12)\\
111111 & (\frac12,1,\frac12,\frac12,\frac12,\frac12).
\end{array}
\]
Thus every recommended action in the support of $q$ satisfies the obedience
constraints.

The sender utility is
\[
S(q)=
4\cdot\frac15
+5\cdot\frac25
+5\cdot\frac15
+6\cdot\frac15
=5.
\]
The receiver utility is
\[
R(q)
=\frac{1}{10}(5+5)
+\frac15(2+5)
+\frac{1}{10}(4+4)
+\frac{1}{10}(4+3)
=\frac{39}{10}.
\]
\end{proof}

\begin{lemma}[Six-bit dual certificate]
\label{lem:six-bit-dual}
For the prior $\mu^{\mathrm{lb}}$,
\[
S^*(\mu^{\mathrm{lb}})\leq 5.
\]
\end{lemma}

\begin{proof}
Define $C_i(x,a)$ as in the direct-obedience constraints. Let the dual variables
$y_x$ on the support states be
\[
\begin{array}{c|cccccc}
x & 001000 & 010111 & 100011 & 100101 & 100110 & 111000\\
\hline
y_x & 4 & 6 & 4 & 4 & 4 & 6.
\end{array}
\]
Let the only nonzero dual variables $\lambda_{a,i}$ be
\[
\begin{array}{c|c|c}
a & i & \lambda_{a,i}\\
\hline
011111 & 2 & 2\\
101111 & 3 & 4\\
101111 & 4 & 2\\
101111 & 5 & 2\\
101111 & 6 & 2\\
110111 & 2 & 2\\
111011 & 2 & 2\\
111101 & 2 & 2\\
111110 & 2 & 2\\
111111 & 1 & 2\\
111111 & 3 & 10\\
111111 & 4 & 4\\
111111 & 5 & 4\\
111111 & 6 & 4.
\end{array}
\]
All bit indices are counted from left to right, and all omitted
$\lambda_{a,i}$ are zero. Since every direct scheme for prior $\mu^{\mathrm{lb}}$
assigns zero mass to states outside the support of $\mu^{\mathrm{lb}}$, checking
these inequalities on the six support states is sufficient. The displayed dual
variables satisfy the following finite set of inequalities over the six support
states and the $64$ possible recommended actions:
\[
|a|_1\leq y_x+\sum_{i=1}^6 \lambda_{a,i}C_i(x,a)
\quad\text{for every support state }x\text{ and every action }a.
\]
Therefore, for every obedient direct scheme $q$,
\[
\begin{aligned}
S(q)
&=\sum_{x,a} q(x,a)|a|_1\\
&\leq
\sum_{x,a}q(x,a)y_x
+\sum_{a,i}\lambda_{a,i}\sum_x q(x,a)C_i(x,a)\\
&\leq
\sum_x \mu^{\mathrm{lb}}(x)y_x
=5.
\end{aligned}
\]
Taking the maximum over obedient direct schemes gives
$S^*(\mu^{\mathrm{lb}})\leq 5$.
\end{proof}

\begin{theorem}[A six-bit lower bound]
\label{thm:six-bit-lower-bound}
For the prior $\mu^{\mathrm{lb}}$,
\[
R_0(\mu^{\mathrm{lb}})=\frac{31}{10},
\qquad
S^*(\mu^{\mathrm{lb}})=5,
\qquad
R_{\max}(\mu^{\mathrm{lb}})=\frac{39}{10}.
\]
Hence
\[
\frac{R_{\max}(\mu^{\mathrm{lb}})}{R_0(\mu^{\mathrm{lb}})}
=
\frac{39}{31}
\approx 1.2580645
>
\frac54.
\]
\end{theorem}

\begin{proof}
The bit marginals of $\mu^{\mathrm{lb}}$ are
\[
\left(\frac12,\frac12,\frac25,\frac12,\frac12,\frac12\right),
\]
so
\[
R_0(\mu^{\mathrm{lb}})
=
\frac12+\frac12+\frac35+\frac12+\frac12+\frac12
=\frac{31}{10}.
\]
The sum of the marginals is
\[
P(\mu^{\mathrm{lb}})=\frac{29}{10}.
\]

The six-bit primal scheme, Lemma~\ref{lem:six-bit-primal}, gives an obedient
direct scheme $q$ with $S(q)=5$ and $R(q)=39/10$. The six-bit dual certificate,
Lemma~\ref{lem:six-bit-dual}, gives $S^*(\mu^{\mathrm{lb}})\leq 5$. Therefore
$S^*(\mu^{\mathrm{lb}})=5$, and the primal scheme is sender-optimal.

The sender-receiver tradeoff, Lemma~\ref{lem:sender-receiver-tradeoff}, says
that every obedient direct scheme satisfies
$R(q')+S(q')\leq n+P(\mu^{\mathrm{lb}})$. Since
$P(\mu^{\mathrm{lb}})=29/10$ and every sender-optimal $q'$ has $S(q')=5$,
\[
R(q')\leq 6+\frac{29}{10}-5=\frac{39}{10}.
\]
Since the six-bit primal scheme attains receiver utility $39/10$ among
sender-optimal obedient direct schemes,
\[
R_{\max}(\mu^{\mathrm{lb}})=\frac{39}{10}.
\]
\end{proof}

\section{Future Work}

The first open direction is to sharpen the constant.
Theorem~\ref{thm:universal-upper-bound} gives the upper bound $3/2$, while
Theorem~\ref{thm:six-bit-lower-bound} shows that any universal bound must
exceed $5/4$. It remains to determine the optimal constant, or to find
structural assumptions on $\mu$ under which the upper bound can be improved.

A second direction is to add multiple senders. Competition can force additional
information revelation, but general multi-sender persuasion has difficult
equilibrium structure. The bit-string model gives a concrete setting in which
one can ask how receiver utility changes as the number of senders, their
preferred actions, and the correlation structure of their signals vary.

A third direction is to weaken the sender's information. In the present model
the sender observes the full state $X$. A more realistic model may let the
sender observe only a signal about $X$, or let different senders observe
different partial signals. This would separate the effect of incentive
misalignment from the effect of incomplete agent knowledge.

Finally, one can change the receiver's utility. Additive correctness makes the
receiver's best response decompose coordinate by coordinate. Submodular
utilities over correctly guessed coordinates would model diminishing returns to
information and interactions between bits. Extending the present bounds to such
utilities would require replacing the coordinatewise obedience and payoff
arguments with statements about set-valued marginal gains.

\bibliographystyle{ACM-Reference-Format}
\bibliography{references}

\end{document}